\documentclass[12pt]{article}

\usepackage[breaklinks]{hyperref}  
\usepackage{graphicx}
\usepackage{xspace}
\usepackage{color}
\usepackage{euscript}
\usepackage{theorem}
\usepackage{amsmath}
\usepackage{enumerate}

\definecolor{blue25}{rgb}{0,0,0.25}

\newcommand{\MakeBig}{\rule[-.2cm]{0cm}{0.4cm}}
\newcommand{\sep}[1]{\,\left|\, {#1} \MakeBig\right.}
\newcommand{\pth}[2][\!]{#1\left({#2}\right)}

\newtheorem{theorem}{Theorem}[section]
\newtheorem{lemma}[theorem]{Lemma}
\newtheorem{claim}[theorem]{Claim}
{\theorembodyfont{\rm} \newtheorem{defn}[theorem]{Definition}}

\newenvironment{proof}{{\em Proof:}}{\hfill{\hfill\rule{2mm}{2mm}}}

\newcommand{\brc}[1]{\left\{ {#1} \right\}}

\newcommand{\List}[1]{\mathcal{L}_{#1}}
\newcommand{\LSet}{\mathsf{Z}}

\newcommand{\Arr}{\EuScript{S}}
\newcommand{\ArrB}{\EuScript{T}}

\newcommand{\KMed}{$k$-range-medians\xspace}
\newcommand{\KMedQ}{median range queries\xspace}

\newcommand{\clmlab}[1]{\label{claim:#1}}
\newcommand{\clmref}[1]{Claim~\ref{claim:#1}}
\newcommand{\seclab}[1]{\label{sec:#1}}
\newcommand{\secref}[1]{Section~\ref{sec:#1}}
\newcommand{\apndref}[1]{Appendix~\ref{sec:#1}}

\newcommand{\lemlab}[1]{\label{lemma:#1}}
\newcommand{\lemref}[1]{Lemma~\ref{lemma:#1}}

\newcommand{\thmlab}[1]{\label{thm:#1}}
\newcommand{\thmref}[1]{Theorem~\ref{thm:#1}}

\newcommand{\rank}{\mathop{\mathrm{rank}}}
\newcommand{\ceil}[1]{\left\lceil {#1} \right\rceil}
\newcommand{\floor}[1]{\left\lfloor {#1} \right\rfloor}

\newcommand{\ncurr}{n_{\mathrm{curr}}}

\newcommand{\cardin}[1]{\left| {#1} \right|}

\newcommand{\junk}[1]{}

\renewcommand{\th}{th\xspace}
\newcommand{\eps}{{\varepsilon}}%

\newcommand{\atgen}{\symbol{'100}}
\newcommand{\SarielThanks}[1]{\thanks{Department of Computer
      Science; 
      University of Illinois; 
      201 N. Goodwin Avenue;
      Urbana, IL, 61801, USA;
      {\tt sariel\atgen{}uiuc.edu}; {\tt
         \url{http://www.uiuc.edu/\string~sariel/}.} #1}}

\newcommand{\MedianAlg}{\Algorithm{MedianAlg}\xspace}
\DefineNamedColor{named}{RedViolet}     {cmyk}{0.07,0.90,0,0.34}
\newcommand{\Algorithm}[1]{{\textcolor[named]{RedViolet}{\texttt{\bf{#1}}}}}
\newcommand{\etal}{\textit{et~al.}\xspace}

\begin{document}

\title{Range Medians}%
\author{Sariel Har-Peled\SarielThanks{}%
   \and%
   S.  Muthukrishnan%
   \thanks{Google Inc., 76 9th Av, 4th Fl., New York,
      NY, 10011. {\tt{muthu\atgen{}google.com}}}}

\date{\today}

\maketitle

\begin{abstract}
    We study a generalization of the classical median finding problem
    to batched query case: given an array of unsorted $n$ items and
    $k$ (not necessarily disjoint) intervals in the array, the goal is
    to determine the median in {\em each} of the intervals in the
    array. We give an algorithm that uses $O(n\log n + k\log k \log
    n)$ comparisons and show a lower bound of $\Omega(n\log k)$
    comparisons for this problem. This is optimal for $k=O(n/\log n)$.
\end{abstract}


\section{Introduction}

The classical median finding problem is to find the {\em median} item,
that is, the item of rank $\lceil n/2\rceil$ in an unsorted array of
size $n$.  We focus on the comparison model, where items in the array
can be compared only using comparisons, and we count the number of
comparisons performed by any algorithm~\footnote{In the algorithms
   discussed in this paper, the computation performed beyond the
   comparisons will be linear in the number of comparisons.}.  It is
known since the 70's that this problem can be solved using $O(n)$
comparisons in the worst case~\cite{bfprt-tbs-73}.  Later research
\cite{bj-fmrtc-85,spp-fm-76,dz-sm-99,dz-msrtc-01} showed that the
number of comparisons needed for solving the median finding algorithm
is between $(2+\eps)n$ and $2.95n$ in the worst case (in the
deterministic case). Closing this gap for a deterministic algorithm is
an open problem, but surprisingly, one can find the median using $1.5n
+ o(n)$ comparisons using a randomized algorithm \cite{mr-ra-95}.

We study the following generalization of the median problem.

\paragraph{The \KMed Problem.}  
The input is an unsorted array $\Arr$ with $n$ entries. A sequence of
$k$ queries $Q_1, \ldots, Q_k$ is provided. A query $Q_j = [l_j, r_j]$
is an {\em interval} of the array, and the output is $x_1, \ldots,
x_k$, where
\[
x_j = {\rm median} \brc{ \MakeBig \Arr[ l_j], \Arr[ l_j +1], \ldots, \Arr[r_j]}
\]
for $j=1, \ldots ,k$. We refer to this as the {\em \KMed problem}.
The problem is to build a data-structure for $\Arr$ such that it can
answer this kind of queries quickly. Notice that the intervals are possibly
overlapping.

\medskip 

This is the {\em interval} version of the classical median finding
problem, and it is interesting on its own merit. In addition, there
are many motivating scenarios where they arise.

\paragraph{Examples.}
A motivation arises in analyzing logs of internet advertisements (aka
ads).  We have the log of clicks on ads on the internet: each record
gives the time of the click as well as the varying price paid by the
advertiser for the click, and the log is arranged in time-indexed
order.  Then, $\Arr[i]$ is the price for the $i$th click.  Any given
advertiser runs several ad campaigns simultaneously spread over
different intervals of time.  The advertiser then wishes to compare
his cost to the general ad market during the period his campaigns ran,
and a typical comparison is to the median price paid for clicks during
those time intervals. This yields an instance of the \KMed problem,
for possibly intersecting set of intervals.

As another example, consider IP networks where one collects what are
known as SNMP logs: for each link that connects two routers, one
collects the total bytes sent on that link in each fixed length
duration like say 5 minutes~\cite{kmz-cbmdq-03}. Then, 
$\Arr[i]$ is the number of bytes sent on that link in the $i$th time
duration. A traffic analyst is interested in finding the median value
of the traffic level within a \emph{specific} time window such as a
week, office hours, or weekends, or the median within \emph{each} such
time window.  Equally, the analyst is sometimes interested in median
traffic levels during specific external events such as the time
duration when an attack happened or a new network routing strategy was
tested.

There are other attributes in addition to time where applications may
solve range median problems. For example, $\Arr[i]$ may be the total
value of real estate sold in postal zipcode area $i$ arranged in
sorted order, and an analyst may be interested in the median value for
a borough or a city represented by a consecutive set of zipcodes.

\medskip
One can ask similar interval versions of other problems too, for
example, the median may be replaced by (say) the maximum, minimum,
mode or even the sum.  
\begin{itemize}
    \item For sum, a trivial $O(n)$ preprocessing to compute all the
    prefix sums $P[j] = \sum_{i\le j} \Arr[i]$ suffices to answer any
    interval query $Q_j=[l_j,r_j]$ in optimal $O(1)$ time using
    $P[r_j]-P[l_j-1]$.

    \item If the summation operator (i.e., $\sum$) is replaced by a
    semigroup operator (where the subtraction operator is absent),
    then $\Arr$ can be preprocessed in $O(nk)$ space and time and each
    query can be answered in $O\pth{ \alpha_k(n) }$ where $\alpha_k$
    is a slow growing function~\cite{y-sttar-82}, and this is optimal
    under general semigroup conditions~\cite{y-cmps-85}.

    \item For the special cases of the semigroup operator such as the
    maximum or minimum, a somewhat nontrivial algorithm is needed to
    get same optimal bounds as for the $\sum$ case (see for
    example~\cite{bf-laps-04}).
\end{itemize}
The median operator is not a semigroup operator and presents a more
difficult problem.  The only prior results we know are obtained by
using the various tradeoffs shown in~\cite{kms-rmrmq-05}.  For the
case when $k=1$, the interesting tradeoffs for preprocessing time and
query times are respectively, roughly, $O(n\log ^2 n)$ and $O(\log
n)$, or $O(n^2)$ and $O(1)$, or $O(n)$ and $O(n^{\varepsilon})$ for
constant fraction $\varepsilon$~\cite{kms-rmrmq-05}.  These bounds for
individual queries can be directly applied to each of the $k$ interval
queries in our problem, resulting in a multiplicative $k$ factor in
the query complexity. In particular, the work of Krizanc \etal
\cite{kms-rmrmq-05} implies an $O\pth{ n \log^2 n + k \log n}$ time
algorithm for our problem.

Our main result is as follows.
\begin{theorem}
    There is a deterministic algorithm to solve the \KMed problem in
    $O( n \log k + k \log k \log n )$ time. Furthermore, in the
    comparison model, any algorithm that solves this problem requires
    $\Omega(n \log k )$ comparisons. 

    \thmlab{main}
\end{theorem}

The \KMed problem seems to be a fairly basic problem and it is worthwhile
to have tight bounds for it. In particular, $\Theta(n\log k)$ may not be the 
bound one suspects at first glance to be tight for this problem.
For $k=O(n/log n)$, our algorithm is optimal. It also improves 
\cite{kms-rmrmq-05} for $k=O(n)$.

The lower bound holds even if the set of intervals is
\emph{hierarchical}, that is, for any two intervals in the set, either
one of them is contained in the other, or they are disjoint. On the
other hand, the upper bound holds even if the queries arrive online,
in the amortized sense. Our algorithm uses {\em relaxed} sorting on
pieces of the array, where only a subset of items in a piece is in
their correct sorted location.  Relaxed sorting like this has been
used before for other problems, for example, see~\cite{AY}.

In the following, the $k$th element of a set $S$ (or element of rank
$k$) would refer to the $k$th smallest element in the set $S$.  For
simplicity, we assume the elements of $\Arr$ are all unique.

\section{The Lower Bound}

Recall that $\Arr$ is an unsorted array of $n$ elements.  Assume that
$n$ is a multiple of $k$.  Let $\Psi(n,k) = \brc{ \frac{i n}{k} \sep{
      i =1,\ldots, k}}$, for $n > k > 0$. We will say an element of
$\Arr$ is the \emph{$i$\th element} of $\Arr$ if its rank in $\Arr$ is
$i$.

\begin{claim}
    Any algorithm \MedianAlg that computes all the elements of rank in
    $\Psi(n,k)$ from $\Arr$ needs to perform $\Omega(n \log k)$
    comparisons in the worst case.

    \clmlab{lower:first}
\end{claim}

\begin{proof}
    Let $m_i = {i n}/{k}$, for $i=0,\ldots, k$.  An element would be
    \emph{labeled} $i$ if it is larger than the $m_{i-1}$\th element
    of $\Arr$ and smaller than the $m_i$\th element of $\Arr$ (note,
    that the $m_k$\th element of $\Arr$ is the largest element in
    $\Arr$). An element would be \emph{unlabeled} if its rank in
    $\Arr$ is in $\Psi(n,k)$.
    
    Note, that the output of the algorithm is the indices of the $k$
    unlabeled elements. We will argue that just computing these $k$
    numbers requires $\Omega(n \log k)$ time.

    Consider an execution of \MedianAlg on $\Arr$. We consider the
    comparison tree model, where the input travels down the decision
    tree from the root, at any vertex a comparison is being made, the
    and the input is directed either to the right or left child
    depending on the result of the comparison. 

    A labelling (at a vertex $v$ of the decision tree) is
    \emph{consistent} with the comparisons seen so far by the
    algorithm if there is an input with this labelling, such that it
    agrees with all the comparisons seen so far and it reaches $v$
    during the execution.  Let $\LSet$ be the set of labellings of
    $\Arr$ consistent with the comparisons seen so far at this vertex
    $v$.

    We claim that if $|\LSet| > 1$ then the algorithm can not yet
    terminate. Indeed, in such a case there are at least two different
    labellings that are consistent with the comparisons seen so
    far. If not all the labellings of $\LSet$ have the same set of $k$
    elements marked as unlabeled, then the algorithm has different
    output (i.e., the output is just the indices of the unlabeled
    elements), and as such the algorithm can not terminate.
    
    So, let $\Arr[\alpha]$ be an element that has two different labels
    in two labellings of $\LSet$.  There exists two distinct inputs $B
    = [b_1, \ldots, b_n]$ and $C= [c_1, \ldots, c_n]$ that realizes
    these two labellings.  Now consider the input $D(t) = [d_1(t),
    \ldots, d_n(t)]$, where $d_i(t) = b_i(1-t) + t c_i$, for
    $t\in[0,1]$ and $i=1,\ldots, n$. We can perturb the numbers $b_1,
    \ldots, b_n$ and $c_1,\ldots, c_n$ so that there is never a
    $t\in[0,1]$ for which three entries of $D(\cdot)$ are equal to
    each other (this can be guaranteed by adding random infinitesimal
    noise to each number, and observing that the probability of this
    bad event has measure zero). Note that $D(0) = B$ and $D(1) = C$.

    Furthermore, since for the inputs $B$ and $C$ our algorithm had
    reached the same node (i.e., $v$) in the decision tree, it holds
    that for all the comparisons the algorithm performed so far, it
    got exactly the same results for both inputs.

    Now, assume without loss of generality, that the label for
    $b_{\alpha}$ in $B$ is strictly smaller than the label for
    $c_{\alpha}$ in $C$.  Clearly, for some value of $t$ in this
    range, denoted by $t^*$, $d_\alpha(t)$ must be of rank in the set
    $\brc{m_1, \ldots, m_k}$. Indeed, as $t$ increases from $0$ to
    $1$, the rank of $d_\alpha(t)$ starts at the rank of $b_\alpha$ in
    $B$, and ends up with the rank of $c_\alpha$ in $C$. But $D(t^*)$
    agrees with all the comparisons seen by the algorithm so far
    (since if $b_i < b_j$ and $c_i < c_j$ then $d_i(t) <d_j(t) $, for
    $t \in [0,1]$). We conclude that the assignment that realizes
    $D(t^*)$ must leave $d_\alpha(t)$ unlabeled. Namely, the set
    $\LSet$ has two labellings with different sets of $k$ elements
    that are unlabeled, and as such the algorithm can not terminate
    and must perform SOME more comparisons if it reached $v$ (i.e.,
    $v$ is not a leaf of the decision tree).

    Thus, the algorithm can terminate only when $|\LSet|=1$. Let
    $\beta = n/ k -1$, and observe that in the beginning of \MedianAlg
    execution, it has
    \[
    M =\frac{n!}{k! \pth{\beta!}^k}
    \]
    possible labellings for the output.  Indeed, a consistent
    labeling, is made out of $k$ unlabeled elements, and then $\beta$
    elements are labeled by $i$, for $i=1, \ldots, k$. Now, by
    Stirling's approximation, we have
    \[
    M \geq \frac{(\beta k)!}{\pth{\beta!}^k} \approx \frac{\sqrt{2 \pi
          \beta k} \frac{(\beta k)^{\beta k}}{{e^{\beta k}}}} {\pth{
          \sqrt{2 \pi \beta} \frac{ \beta^\beta}{e^\beta}}^k }
    \frac{(\beta k)!}{\pth{\beta!}^k} %
    =%
    \frac{\sqrt{2 \pi \beta k} {(\beta k)^{\beta k}}{{}}} {\pth{
          \sqrt{2 \pi \beta}{ \beta^\beta}}^k } =
    k^{\beta k } 
    \frac{\sqrt{2 \pi \beta k} } {\pth{
          \sqrt{2 \pi \beta}}^k}.
    \]
    Each comparison performed can only half this set of possible
    labellings, in the worst case. It follows, that in the worst case,
    the algorithms needs 
    \[
    \Omega\pth{ \log M } = \Omega\pth{ \beta k \log k - \frac{k}{2}
       \log \pth{2 \pi \beta}} = \Omega\pth{ \beta k \log k} =
    \Omega\pth{ n \log k }
    \]
    comparisons, as claimed.
\end{proof}

\begin{lemma}
    Solving the \KMed problem requires $\Omega(n \log k )$ comparisons.
\end{lemma}
\begin{proof}
    We will show that given an algorithm for the \KMed problem, one
    can reduce it, in linear time, to the problem of
    \clmref{lower:first}. That would immediately imply the lower
    bound.

    Given an input array $\Arr$ of size $n$, construct a new
    array $\ArrB$ of size $4n$ where the first $n$ elements of
    $\ArrB$ are $- \infty$, $\ArrB[n+1, \ldots, 2n ] =\Arr$, and
    $\ArrB[j] = +\infty$, for $j =2n+1, \ldots, 4n$. Clearly, the
    $\ell$\th element of $\Arr$ is the median of the range $[1, 2n
    +2\ell-1]$ in $\ArrB$. Thus, we can solve the problem of
    \clmref{lower:first} using $k$ \KMedQ, implying the lower bound.
\end{proof}

\medskip
Observe that the lower bound holds even for the case when the intervals
are hierarchical. 

\section{Our Algorithm}

We first consider the case when all the query intervals are provided
ahead of time.  We will present a slow algorithm first, and later show
how to make it faster to get our bounds.  Our algorithm uses
the following folklore result.

\begin{theorem}
    Given $\ell$ sorted arrays with total size $n$, there is a
    deterministic algorithm to determine median of the set formed by
    the union of these arrays using $O( \ell \log (n/\ell) )$
    comparisons.

    \thmlab{sorted:median:2}
\end{theorem}

Since we were unable to find a reference to precisely this result 
beyond~\cite{kms-rmrmq-05} where a slightly weaker result is stated as a folklore
claim, we describe this algorithm in \apndref{deterministic:algorithm}.

\subsection{A Slow Algorithm}
\seclab{slow}

Here we show how to solve the \KMed problem.

Let $I_1, \ldots, I_k$ be the given (not necessarily disjoint) $k$
intervals in the array $\Arr[1..n]$. We break $\Arr$ into (at most)
$2k-1$ atomic disjoint intervals labeled in the sorted order $B_1,
\ldots, B_m$, such that an atomic interval does not have an endpoint
of any $I_i$ inside it. Next, we sort each one of the $B_i$'s, and
build a balanced binary tree having $B_1, \ldots, B_m$ as the leaves
in this order. In a bottom-up fashion we merge the sorted arrays
sorted in the leaves, so that each node $v$ stores a sorted array
$S_v$ of all the elements stored in its subtree.  Let $T$ denote this
tree that has height $O( \log k)$.

Now, computing the median of an interval $I_j$, is done by extracting
the $O( \log k )$ suitable nodes in $T$ that cover $I_j$. Next, we
apply \thmref{sorted:median:2}, and using $O( \log n \log k )$
comparisons, we get the desired median. We now apply this to the $k$
given intervals.  Observe that sorting the atomic intervals takes $O(
n \log n)$ comparisons and merging them in $O( \log k )$ levels takes
$O( n \log k)$ comparisons in all. This gives:

\begin{lemma}
    The algorithm above uses $O( n \log n + k \log n \log k )$
    comparisons.
\end{lemma}

Note, that this algorithm is still mildly interesting. Indeed, if the
intervals $I_1, \ldots, I_k$ are all ``large'', then the running time
of the naive algorithm is $O(n k)$, and the above algorithm is faster
for $k > \log n$.

\subsection{Our Main Algorithm}
\seclab{faster:algorithm}

The main bottleneck in the above solution was the presorting of the
pieces of the array corresponding to atomic intervals.  In the optimal
algorithm below, we do not fully sort them.

\begin{defn}
    A subarray $X$ is \emph{$u$-sorted} if there is a sorted list
    $\List{X}$ of at most (say) $20 u$ elements of $X$ such that these
    elements appear in this sorted order in $X$ (not necessarily as
    consecutive elements). Furthermore, for an element $\alpha$ of
    $\List{X}$, all the elements of $X$ smaller than it appear before
    it in $X$ and all the elements larger than $\alpha$ appear after
    $\alpha$ in $X$.  Finally, we require that the distance between
    two consecutive elements of $\List{X}$ in $X$ is at most $|X|/u$,
    where $|X|$ denotes the size of $X$.  We will refer to the
    elements of $X$ between two consecutive elements of $\List{X}$ as
    a \emph{segment}.
\end{defn}

An array $X$ of $n$ elements that is $n$-sorted is just
sorted, and a $0$-sorted array is unsorted. Another way to look at it,
is that the elements of $\List{X}$ are in their final position in the
sorted order, and the elements of the intervals are in an arbitrary
ordering.

\begin{lemma}
    Given an unsorted array $X$, it can be $u$-sorted using $O(|X| \log u
    )$ comparisons, where $|X|$ denotes the number of elements of $X$.
\end{lemma}

\begin{proof}
    We just find the median of $X$, partition $X$ into two equal size
    subarrays, and continue recursively on the two subarrays. The
    depth the recursion is $O( \log u)$, and the work at each level of
    the recursion is linear, which implies the claim.
\end{proof}

\begin{lemma}
    Given a two $u$-sorted arrays $X$ and $Y$, they can be merged into
    an $u$-sorted array using $O( |X| + |Y|)$ comparisons.
\end{lemma}

\begin{proof}
    Convert $Y$ into a linked list. Insert the elements of $\List{X}$
    into $Y$. This can be done by scanning the list of $Y$ until we
    arrive at the segment $Y_i$ of $Y$ that should contain an
    element $b$ of $\List{X}$ that we need to insert. We partition
    this segment using $b$ into two intervals, add $b$ to
    $\List{Y}$, and continue in this fashion with each such $b$.  
    This takes $O(|Y_i|) =
    O(|Y|/u)$ comparisons per $b$ (ignoring the scanning cost which is $O(|Y|)$
    overall).  Let $Z$ be the resulting $u$-sorted array, which
    contains all the elements of $Y$ and all the elements of
    $\List{X}$, and $\List{Z} = \List{X} \cup \List{Y}$.  Computing
    $Z$ takes
    \[
    O\pth{ |Y| + |\List{X}| \frac{|Y|}{u} } = O( |Y| )
    \]
    comparisons. 

    We now need to insert the elements of $X \setminus \List{X}$ into
    $Z$. Clearly, if a segment $X_i$ of $X$ has $\alpha_i$ elements of
    $\List{Z}$ in its range, then inserting the elements of $X_i$
    would take $O(|X_i| \log \alpha_i)$ comparisons. Thus, the total
    number of comparisons is
    \[
    O \pth { \sum_{i} |X_i| \log \alpha_i } = O\pth{ \sum_i
       \frac{\cardin{X}}{u} \log \alpha_i } = O\pth{
       \frac{\cardin{X}}{u} \sum_i \alpha_i } = O\pth{ X },
    \]
    since $|X_i| \leq |X|/u$, $\log \alpha_i \leq \alpha_i$ and
    $\sum_{i} \alpha_i = O(u)$.

    The final step is to scan over $Z$, and merge consecutive
    intervals that are too small (removing the corresponding elements
    from $\List{Z}$), such that each interval is of length at most
    $|Z|/u$. Clearly, this can be done in linear time. The resulting
    $Z$ is $u$-sorted since its sorted list contains at most $2u +1$
    elements, and every interval is of length at most $|Z|/u$. 
\end{proof}

\medskip

Note, that the final filtering stage in the above algorithm is need to
guarantee that the resulting list $\List{Z}$ size is not too large, if
we were to use this merging step several times.

\medskip

In the following, we need a modified version of
\thmref{sorted:median:2} that works for $u$-sorted arrays.
\begin{theorem}
    Given $\ell$ $u$-sorted arrays $A_1, \ldots, A_\ell$ with total
    size $n$ and a rank $k$, there is a deterministic algorithm that
    returns $\ell$ subintervals $B_1, \ldots, B_\ell$ of
    these arrays and a number $k'$, such that the following properties
    hold.
    \begin{enumerate}[(i)]
        \item The $k'$\th ranked element of $B_1 \cup \cdots \cup
        B_\ell$ is the $k$\th ranked element of $A_1 \cup \cdots \cup
        A_\ell$.

        \item The running time is  $O( \ell \log (n/\ell) )$ time.

        \item $\sum_{i=1}^\ell \cardin{B_i} = O\pth{ \ell \cdot (n/u) }$.
    \end{enumerate}

    \thmlab{sorted:median:3}
\end{theorem}

\begin{proof}
    For every element of $\List{A_i}$ realizing the $u$-sorting of the
    array $A_i$, we assume we have its rank in $A_i$ precomputed.
    Now, we execute the algorithm of \thmref{sorted:median:2} on these
    (representative) sorted arrays (taking into account their
    associated rank). (Note that the required modifications of the algorithm of
    \thmref{sorted:median:2} are tedious but straightforward, and we
    omit the details.) The main problem is that now the rank of an
    element is only estimated approximately up to an (additive) error
    of $n/u$. In the end of process of trimming down the
    representative arrays, we might still have active intervals of
    total length $2n/u$ in each one of these arrays, resulting in the
    bound on the size of the computed intervals.
\end{proof}

Using the theorem above as well as two lemmas above, we get the
following result, which is building up to the algorithmic part of
\thmref{main}.

\begin{lemma}
    There is a deterministic algorithm to solve the \KMed problem in
    $O( n \log k + k \log k \log n )$ time, when the $k$ query
    intervals are provided in advance.

    \lemlab{main}
\end{lemma}

\begin{proof}
    We repeat the algorithm of \secref{slow} using $u$-sorting instead
    of sorting, for $u$ to be specified shortly. Building the
    data-structure (i.e., the tree over the atomic intervals) takes
    $O( n \log u)$ comparisons. Indeed, we first $u$-sort the atomic
    intervals, and then we merge them as we go up the tree.

    A query of finding the median of array elements in an interval is
    now equivalent to finding the median for $m = O(\log k)$
    $u$-sorted arrays $A_1, \ldots, A_m$. Using the algorithm of
    \thmref{sorted:median:3} results in $m$ intervals $B_1, \ldots,
    B_m$ that belong to $A_1, \ldots, A_m$, respectively, such that we
    need to find the $k'$\th smallest element in $B_1 \cup \ldots \cup
    B_m$. The total length of the $B_i$s is $O( m n/u)$.  Now we can
    just use the brute force method. Merge $B_1, \ldots, B_m$ into a
    single array and find the $k'$\th smallest element using the
    classical algorithm. This take $ O( m n/u )$ comparisons. We have
    to repeat this $k$ times, and the number of comparisons we need is
    \[
    O\pth{ k m \frac{n}{u} + k m \log n } = O(n + k \log k \log n ),
    \]
    for $u = k^2$, since $m= O(\log k)$.  Thus, in all, the number of
    comparisons using by the algorithm is $O(n \log k + k\log k \log
    n)$.
\end{proof}

\medskip 

We can extend this bound to the case when the intervals are presented
in an online manner, and we get amortized bounds.

\begin{lemma}[When $k$ is known in advance.]
    There is a deterministic algorithm to solve the \KMed problem in
    $O( n \log k + k \log k \log n )$ time, when the $k$ query
    intervals are provided in an online fashion, but $k$ is known in
    advance.

    \lemlab{main:2}
\end{lemma}

\begin{proof}
    The idea is to partition the array into $u$, $u\leq k^2$ atomic intervals
    all of the same length, and build the data-structure of these
    atomic intervals. The above algorithm would work verbatim, except
    for every query interval $I$, there would be two ``dangling''
    atomic intervals that are of size $n/u$ that contain the two
    endpoints of $I$. 

    Specifically, to perform the query for $I$, we compute $m = O(\log
    k)$ $u$-sorted arrays using our data-structure.  We also take
    these two atomic intervals, clip them into the query interval,
    $u$-sort them, and add them to the $m$ $u$-sorted arrays we
    already have. Now, we need to perform the median query over these
    $O( \log k)$ $u$-sorted arrays, which we can do, as described above.
    Clearly, the resulting algorithm has running time
    \[
    O\pth{ n \log u + k \log u \log n + k \frac{n}{u} \log u } =
    O\pth{ n \log k + k \log k \log n},
    \]
    since $u= k^2$.
\end{proof}

\begin{lemma}[When $k$ is \emph{not} known in advance.]
    There is a deterministic algorithm to solve the \KMed problem in
    $O( n \log k + k \log k \log n )$ time, when the $k$ query
    intervals are provided in an online fashion.

    \lemlab{main:3}
\end{lemma}
\begin{proof}
    We will use the algorithm of \lemref{main:2}.

    At each stage, we have a current guess to the number of queries to
    be performed. In the beginning this guess is a constant, say
    10. When this number of queries is exceeded, we \emph{square} our
    guess, rebuild our data-structure from scratch for this new guess,
    and continue. Let $k_1 = 10$ and $k_i = (k_{i-1})^2$ be the
    sequence of guesses, for $i=1,\ldots, \beta$, where $\beta =
    O(\log \log k )$. We have that the total running time of the
    algorithm is
    \[
    \sum_{i=1}^\beta O\pth{ n \log k_i + k_i \log k_i \log n}
    =O( n \log k + k \log k \log n ),
    \]
    since $\log k_{i-1} = (\log k_i)/2$, for all $i$.
\end{proof}

\medskip

\lemref{main:3} implies the algorithmic part of \thmref{main}.


\section{Concluding Remarks}

The \KMed problem is a natural interval generalization of the
classical median finding problem: unlike interval generalizations of
other problems such as $\max$, $\min$ or sum which can be solved in
linear time, our problem (surprisingly) needs $\Omega(n\log k)$
comparisons, and we present an algorithm that solves this problem with
running time (and number of comparisons) $O(n \log k + k\log k \log
n)$. A number of technical problems remain and we list them below.
\begin{itemize}
    \item Currently, our algorithm uses $O(n \log k)$ space. It would
    be interesting to reduce this to linear space.

    \item Say the elements are from an integer range $1,\ldots,U$. Can
    we design $o(n)$ time algorithms in that case using word
    operations? For the classical median finding problem, both
    comparison-based and word-based algorithms take $O(n)$ time. But
    given that the comparison-based algorithm needs $\Omega(n\log k)$
    comparisons for our \KMed problem, it now becomes interesting if
    word-based algorithms can do better for integer alphabet.

    \item Say one wants to only answer median queries approximately
    for each interval (see \cite{bkmt-armrm-05} for some relevant
    results). Can one design $o(n\log k)$ algorithms?

    Suppose the elements are integers in the range $1,\ldots,U$.  We
    define an approximate version where the goal is to return an
    element within $(1\pm\varepsilon)$ of the correct median in value,
    for some fixed $\varepsilon$, $0 < \varepsilon < 1$.  Then we can
    keep an \emph{exponential histogram} with each atomic interval of
    the number of elements in the range
    $[(1+\varepsilon)^i,(1+\varepsilon)^{i+1})$ for each $i$, and
    follow the algorithm outline here constructing them for all the
    suitably chosen intervals on the balanced binary tree atop these atomic
    intervals. For each interval in the query, one can easily merge
    the exponential histograms corresponding to and obtain an
    algorithm that takes time $O(n+k\log k\log U)$, since any two
    exponential histograms can be merged in $O(\log U)$ time.  If the
    elements are not integers in the range $1,\ldots,U$ and one worked
    in the comparison model, similar results may be obtained
    using~\cite{gk-seocq-01,gk-pccos-04}, or $\eps$-nets. It is not
    clear if these bounds are optimal.

    \item We believe extending the problem to two (or more) dimensions
    is also of interest.  There is prior work for range sum and
    minimums, but tight bounds for $k$ range medians will be
    interesting.
\end{itemize}

\medskip
\paragraph{Acknowledgements.} 
The authors would like to thank the anonymous referees for their
careful reading, useful comments and references.  In particular, they
identified mistakes in an earlier version of this paper.

 
\newcommand{\etalchar}[1]{$^{#1}$}

\appendix

\section{Choosing median from sorted
   arrays}
\seclab{deterministic:algorithm}

In this section, we prove \thmref{sorted:median:2} by providing a fast
deterministic algorithm for choosing the median element of $\ell$
sorted arrays. As we mentioned before, this result seems to be known,
but we are unaware of a direct reference to it, and as such we provide
a detailed algorithm.

\subsection{The algorithm}

Let $A_1, \ldots, A_\ell$ be the given sorted arrays of total size
$n$.  We maintain $\ell$ active ranges
$[l_i,r_i]$ of the array $A_i$ where the required element (i.e.,
``median'') lies, for $i=1, \ldots, \ell$. Let $k$ denote the rank of
the required median.  Let $\ncurr = \sum_i (r_i -l_i+1)$ be the total
number of currently active elements.

If $\ncurr \leq 32 \ell$, then we find the median in linear time, using
the standard deterministic algorithm.  Otherwise, let $\Delta =
\floor{\ncurr/(32 \ell)}$. Pick $u_i-1$ equally spaced elements from
the active range of $A_i$, where
\[
u_i= 4+ \ceil{ \frac{ r_i - l_i +1 }{\Delta} }.
\]
Let $L_i$ be the resulting list of representatives, for $i=1,\ldots,
\ell$. Note that $L_i$ breaks the active range of $A_i$ into blocks of
size 
\[
\nu_i \leq \ceil{\frac{r_i - l_i +1}{u_i}}.
\]
For each element of $L_i$ we know exactly how many elements are
smaller than it and larger than it in the $i$\th array. Merge the
lists $L_1,\ldots, L_\ell$ into one sorted list $L$. For an element
$x$, let $\rank(x)$ denote the rank of $x$ in the set $A_1 \cup \ldots
\cup A_\ell$.  Note, that now for every element $x$ of $L$ we can
estimate its $\rank(x)$ to lie within an interval of length $T =
\sum_{i=1}^{\ell} \nu_i$.  Indeed, we know for an element of $x \in L$
between what two consecutive representatives it lies for all $\ell$
arrays.  For element $x \in L$, let $R(x)$ denote this range where
the rank of $x$ might lie.

Now, given two consecutive representatives $x$ and $y$ in the $i$\th
array, if $k \notin R(x)$ and $k \notin R(y)$ then the required median
cannot lie between $x$ and $y$, and we can shrink the active range
not to include this portion. In particular, the new active range spans
all the blocks which might contain the median. The algorithm now
updates the value of $k$ and continues recursively on the new active
ranges.

\subsection{Analysis}

The error estimate for the rank of a representative is bounded by
\begin{eqnarray*}
    U &=& \sum_{i=1}^{\ell} \nu_i
    \leq \ell +  \sum_{i=1}^{\ell} \frac{r_i -
       l_i+1}{u_i} 
    \leq 
    \ell +  \sum_{i=1}^{\ell} \frac{r_i -
       l_i+1}{4+ { \frac{ r_i - l_i +1 }{\Delta} }} 
    = \ell +  \Delta\sum_{i=1}^{\ell} \frac{r_i -
       l_i+1}{4\Delta +  r_i - l_i +1 }
    \\
    &\leq& \ell + \ell \Delta 
    \leq \frac{\ncurr}{32} + 
    \ell \floor{\frac{\ncurr}{32 \ell}}
    \leq  \frac{\ncurr}{16},
\end{eqnarray*}
since $\ell \leq \ncurr / 32$ and by the choice of $\Delta$.

Consider the sorted merged array $B$ of all the active elements. The
length of $B$ is $\ncurr$, and assume, for the sake of simplicity of
exposition, that the desired median is in the second half of $B$ (the
other case follows by a symmetric argument). Note, that any
representative $x$ that fall in the first quarter of $B$ has a rank
that lies in a range shorter than $T < \ncurr/4$, and as such it
cannot include $k$. In particular, let $t_i$ be the index in $A_i$ of
the first representative in the active range (of $A_i$) that does not
falls in the first quarter of $B$. Observe that
$\sum_i (t_i -l_i + 1) \geq \ncurr/4$.
The total number of elements that are being eliminated by the
algorithm (in the top of the recursion) is at least
\[
\sum_i \pth{ (t_i -l_i + 1) - 2 \nu_i } \geq \sum_i \pth{ t_i -l_i +
   1} - 2\sum_i \nu_i = \frac{\ncurr}{4} - 2U \geq \frac{\ncurr}{8}.
\]
Namely, each recursive call continues on total length of all active
ranges smaller by a factor of $(7/8)$ from the original array.

The total length of $L_1 ,\ldots L_\ell$ is $O(\ell)$, and as such the
total work (ignoring the recursive call) is bounded by $O( \ell \log
\ell )$. The running time is bounded by 
\[
T( \ncurr ) = O( \ell \log \ell ) + T\pth{ (7/8) \ncurr  },
\]
where $T(\ell) = O(\ell \log \ell)$.  Thus, the total running time is
$O( \ell \log \ell \log (\ncurr/\ell) )$.

\subsection{Doing even better - a faster algorithm}

Observe, that the bottleneck in the above algorithm is the merger of
the representative lists $L_1, \ldots, L_\ell$. Instead of merging
them, we will compute the median $x$ of $L = L_1 \cup \ldots \cup
L_\ell$. If $R(x)$ does not contain $k$, then we can throw away at
least $\ncurr/4$ elements in the current active ranges and continue
recursively. Otherwise, compute the element $z$ of rank $\ncurr/4$ in
$L$. Clearly, $k \notin R(z)$ and one can throw, as above, as constant
fraction of the active ranges. The resulting running time (ignoring
the recursive call) is $O(\ell)$ (instead of $O( \ell \log \ell
)$). Thus, the running time of the resulting algorithm is $O( \ell
\log (\ncurr/\ell) )$.


\end{document}